\documentclass [journal,onecolumn,11pt]{IEEEtran}
\usepackage{amsfonts,amsmath,amssymb}
\usepackage{indentfirst, setspace}
\usepackage{url,float}
\usepackage{color}
\usepackage[a4paper,ignoreall]{geometry}
\usepackage{longtable}
\usepackage{graphicx}
\usepackage[a4paper,ignoreall]{geometry}
\usepackage{multicol}
\usepackage{stfloats}
\usepackage{enumerate}
\usepackage{cite}
\usepackage{amsthm}
\usepackage{multirow}
\usepackage{amssymb}
\usepackage[square, comma, sort&compress, numbers]{natbib}

\def\qu#1 {\fbox {\footnote {\ }}\ \footnotetext { From Qu: {\color{red}#1}}}
\def\hqu#1 {}
\def\kq#1 {\fbox {\footnote {\ }}\ \footnotetext { From KangQuan: {\color{blue}#1}}}
\def\hkq#1 {}

\newtheorem{Th}{Theorem}[section]
\newtheorem{Cor}[Th]{Corollary}

\newtheorem{Lemma}[Th]{Lemma}
\newtheorem{Def}[Th]{Definition}

\newtheorem{Conj}[Th]{Conjecture}

\newcommand{\gf}{{\mathbb F}}

\makeatletter
\newcommand{\figcaption}{\def\@captype{figure}\caption}
\newcommand{\tabcaption}{\def\@captype{table}\caption}
\makeatother

\begin{document}
	\title{New Permutation Trinomials Constructed from Fractional Polynomials}

\author{{Kangquan Li, Longjiang Qu, Chao Li and Shaojing Fu**}
         \thanks{Kangquan Li, Longjiang Qu and Chao Li are with the College of Science,
          National University of Defense Technology, Changsha, 410073, China.
          Shaojing Fu is with the College of Computer Science, National University of Defense Technology, Changsha, 410073, China.
          E-mail: a940672099@qq.com, ljqu\_happy@hotmail.com, lichao\_nudt@sina.com, shaojing1984@163.com.
          This work is supported by  the National Basic Research Program of China (Grant No. 2013CB338002),
           the Nature Science Foundation of China (NSFC) under Grant 61272484, 11531002, 61572026,
          the Program for New Century Excellent Talents in University (NCET)
           and the Basic Research Fund of National University of Defense Technology (No.CJ 13-02-01). }}

	\maketitle{}

	\begin{abstract}
		Permutation trinomials over finite fields consititute an active research due to their simple algebraic form, additional extraordinary properties and their wide applications in many areas of science and engineering. In the present paper,  six new classes of permutation trinomials over finite fields of even characteristic are constructed from six fractional polynomials. Further, three classes of permutation trinomials over finite fields of characteristic three  are raised. {Distinct from most of the known permutation trinomials which are with fixed exponents, our results are  some general classes of permutation trinomials  with one parameter in the exponents.} Finally, we propose a few conjectures.
	\end{abstract}

\section{Introduction}
Let $\mathbb{F}_q$ be the finite field with $q$ elements. A polynomial $f\in\mathbb{F}_q[x]$ is called a \emph{permutation polynomial} (PP) if the induced mapping $x\to f(x)$ is a permutation of $\mathbb{F}_q$. The study about permutation polynomials over finite fields attracts people's interest for many years due to their wide applications in coding theory \cite{LC,ST}, cryptography \cite{MN,QY,RSL} and combinatorial designs \cite{DJ}. For example, Ding et al.\cite{DJ} constructed a family of skew Hadamard difference sets via the Dickson permutation polynomial of order five, which disproved the longstanding conjecture on skew Hadamard difference sets. More recent progress on permutation polynomials can be found in \cite{SM,CD,HD1,XH,XH3,XH2,XH4,MZ,ZXLC,Zieve}.

Permutation trinomials over finite fields are in particular interesting for their simple algebraic form and additional extraordinary properties. For instances, a certain permutation trinomial of $\gf_{2^{2m+1}}$ was a major  ingredient in the proof of the Welch conjecture  \cite{HD}. The discovery of {another} class of permutation trinomials by Ball and Zieve \cite{SM} provided a way to prove the construction of the Ree-Tits symplectic spreads of $\mathbf{PG}(3,q)$. Hou acquired a necessary and sufficient condition about determining a special permutation trinomial, which is $ax+bx^q+x^{2q-1}\in\mathbb{F}_{q^2}[x]$. More results about permutation trinomials before 2015 could be referred in \cite[Ch.4]{XH1}. In a recent paper \cite{DQ},  Ding, Qu, Wang, et.al. discovered several new classes of permutation trinomials with nonzero trivial coefficients over finite fields with even characteristic. Moreover, a few new classes of permutation polynomials were given by Li, Qu and Chen \cite{LQC}. Two classes of permutation {trinomials} were found by Ma, et.al. \cite{MZ} as well.

{ The PPs with the form {$x^rh(x^{(q-1)/d})$} over $\gf_q$ are interesting and have been  paid particular attention. There is a connection between the PPs of this type and certain permutations of the subgroup of order $d$ of $\gf_q^\ast$. }
 Throughout this paper, for integer $d>0$, $\mu_d=\{x\in\mathbb{\bar{F}}_q : x^d=1\}$, where $\mathbb{\bar{F}}_q$ denotes the algebraic closure of $\mathbb{F}_q$.
\begin{Lemma}
	\cite{PL,WQ,Zieve}
	\label{lem1}
	Pick $d,r > 0$ with $d\mid (q-1)$, and let $h\in\mathbb{F}_q[x]$. Then $f(x)=x^rh(x^{\left.(q-1)\middle/d\right.})$ permutes $\mathbb{F}_q$ if and only if both
	\begin{enumerate}[(1)]
		\item $\mathrm{gcd}(r,\left.(q-1)\middle/d\right.)=1$ and
		\item $x^rh(x)^{\left.(q-1)\middle/d\right.}$ permutes $\mu_d$.
	\end{enumerate}
\end{Lemma}
   Let $r,d,h$ be defined as above. According to Lemma \ref{lem1}, to determine if $x^rh(x^{\left.(q-1)\middle/d\right.})$ permutes $\mathbb{F}_q$, the key point is to prove the condition (2) in Lemma \ref{lem1}. This is usually difficult. However, in various special sitiuations, conditions have been found for $x^rh(x)^{\left.(q-1)\middle/d\right.}$ to permute $\mu_d$. When $h(x)^{\left.(q-1)\middle/d\right.}=x^n$ for all $x\in\mu_d$, it is easy to find this condition  \cite{AW}.
    Let $q=q_0^m$, where $q_0 \equiv 1 \pmod d$ and $d \mid m$, and $h\in \mathbb{F}_{q_0}[x]$. Akbary and Wang \cite{AW}, Laigle-Chapug \cite{LC} proved that $x^rh(x)^{\left.(q-1)\middle/d\right.}$ permutes $\mathbb{F}_q$ if and only if $\mathrm{gcd}(r+n,d)=\mathrm{gcd}(r,\left.(q-1)\middle/d\right.)=1$.
   Zieve made important contributions to determining permutation polynomials with this form. In \cite{Zieve}, Zieve obtained a necessary and sufficient condition about a complex form $h(x)=h_k(x)^t\hat{h}(h_l(x)^{d_0})$, where $h_k=1+x+\cdots+x^{k-1}$ and $t,d_0,\hat{h}$ satisfy some conditions.
   {Moreover,  Zieve found a novel method to  verify that $x^rh(x)^{q-1}$ permutes $\mu_{q+1}$,
    which is the key step to consider whether $x^rh(x^{q-1})$ permutes $\mathbb{F}_{q^2}$.
    The main idea is to exhibit the permutations of $\mu_{q+1}$ and the bijections $\mu_{q+1}\to\mathbb{F}_q\cup\{\infty\}$
    induced by degree-one rational functions. 
   By this method, Zieve obtained some results about  $f(x)=x^rh(x^{q-1})$ in $\mathbb{F}_{q^2}[x]$, where $f(x)=x^{n+k(q+1)}[(\gamma x^{q-1}-\beta)^n-\gamma(x^{q-1}-\gamma^q\beta)^n]$ \cite{Zieve1} or $f(x)=x^{n+k(q+1)}[(\eta x^{q-1}-\beta\eta^q)^n-\eta(x^{q-1}-\beta)^n]$ \cite{Zieve1}. 
    In  \cite[Corollary 1.4]{Zieve1}, there was one class of permutation trinomial (see Theorem \ref{Th1} in Section 2), which included two permutation trinomials raised as conjectures by Tu, Zeng, Hu and Li       \cite{ZXLC}.}

  { In this paper, we combine the above two approaches to investigate  permutation trinomials with the form {$x^rh(x^{(q-1)/d})$}.
More precisely, we mainly consider the permutation trinomials over $\mathbb{F}_{q^2}$ of the form $x^rh(x^{q-1})${. And we call $g(x)=x^rh(x)^{q-1}$ a \emph{fractional polynomial}.}
  Several new classes of permutation trinomials over $\mathbb{F}_{2^{2k}}$ and {$\mathbb{F}_{3^{2k}}$} with such form are constructed.

  For the characteristic $2$ case, let $q=2^k$. From Lemma \ref{lem1},  $f(x)=x^rh(x^{q-1})\in\mathbb{F}_{q^2}$, where $h(x)=1+x^m+x^n  (1<m<n) $, permutes $\mathbb{F}_{q^2}$ if and only if $\mathrm{gcd}(r,q-1)=1$ and $g(x)=x^rh(x)^{q-1}$ permutes $\mu_{q+1}$.
   For  $x\in \mu_{q+1}$, {we first prove $h(x)\neq0$. Then} we have   $$g(x) = x^rh(x)^{q-1} = x^r\frac{h(x)^q}{h(x)} = x^r \frac{1+x^{mq}+x^{nq}}{1+x^m+x^n}= x^r \frac{1+x^{-m}+x^{-n}}{1+x^m+x^n} = x^{r-n} \frac{x^n+x^{n-m}+1}{1+x^m+x^n}. $$

    We find several classes of fractional polynomials $g(x)$ that permute $\mu_{q+1}$.
   Hence more permutation trinomials can be constructed by Lemma \ref{lem1}. In our proofs,   the key point is to determine  the number of the solutions of some specified equations with high degree (cubic, quartic, and so on).

For the characteristic $3$ case, we use another approach since it is difficult to
  prove  that the corresponding  fractional polynomials permutes $\mu_{q+1}$ directly.
  To the authors' best knowledge, this approach was first used by Hou \cite{XH3}.
  The main idea is as follows.  It is clear that a polynomial {\bf $f(x)\in\mathbb{F}_{q}$} permutates $\mathbb{F}_{q^2}$ if and only if $f(x)=c$ has one unique solution in $\mathbb{F}_{q^2}$ for any $c\in\mathbb{F}_{q^2}$. It is usually easy to prove this for the case $c\in\mathbb{F}_q$.
   For the other case  $c\in \mathbb{F}_{q^2}\backslash \mathbb{F}_{q}$, we first prove that $f(\mathbb{F}_{q^2} \backslash \mathbb{F}_q) \subseteq \mathbb{F}_{q^2} \backslash \mathbb{F}_q$. Denote by $\mathrm{Tr}(x)$ and $\mathrm{N}(x)$ the trace function and the norm function from $\mathbb{F}_{q^2}$ to $\mathbb{F}_{q}$, respectively. Then it suffices to show that $\mathrm{Tr}(x)$ and $\mathrm{N}(x)$ are uniquely determined by $c$ since the pair $(\mathrm{Tr}(x), \mathrm{N}(x))$ can be uniquely
    determined the set $\{x, x^q\}$  and $f(x)\neq f(x^q)=f(x)^q$ (We assume $f\in \gf_q[x]$).
    We use this approach to directly prove the permutation properties of two classes of trinomials.
     Further, by Lemma \ref{lem1} we can obtain more permutation trinomials.
  }

The remainder of this paper is organized as follows. In Section 2,  six new classes of permutation trinomials over finite fields with even characteristic are proposed. In Section 3, we {raise} {two} classes of permutation trinomials over finite fields with characteristic three. Section 4 is a comparison between our results and known permutation trinomials. Section 5 are the {conclusion} and some conjectures.

\section{Some permutation trinomials over $\mathbb{F}_{2^{2k}}$}

In this section, we obtain six new classes of permutation trinomials. The first three classes are generated  from two known theorems. The remaining results are constructed from three new fractional polynomials.

The following lemmas will be needed later.

\begin{Lemma}
	\cite{LN}
	\label{lem6}
	Let $q=2^k$, where $k$ is a positive integer. The quadratic equation $x^2+ux+v=0$, where $u,v\in\mathbb{F}_{q}$ and $u\neq 0$, has roots in $\mathbb{F}_q$ if and only if $\mathrm{Tr}_q(\left.{v}\middle/{u^2}\right.)=0$.
\end{Lemma}

\begin{Lemma}
	\cite{SZM}
	\label{lem5}
	Let $a, b \in \mathbb{F}_{q}$, where $q=2^k$ and $b\neq0$. Then the cubic equation $x^3+ax+b=0$ has a unique solution in $ \mathbb{F}_{q}$ if and only if $\mathrm{Tr}_q\left(\frac{a^3}{b^2}+1\right) \neq 0$.
\end{Lemma}

{Throughout this section, we assume that $q=2^{k}$.}

\subsection{Generalization of two known theorems}
First, we conduct a generalization from some existing results. The following result was due to Zieve.
\begin{Th}
	\label{Th1}
	\cite[Corollary 1.3]{Zieve1}
	Let $q$ be a prime power, and let $l$ be a nonnegative integer. The polynomial $f(x)=x^{lq+l+3}+3x^{(l+1)q+l+2}-x^{(l+3)q+l}$ permutes $\mathbb{F}_{q^2}$ if and only if $\mathrm{gcd}(2l+3,q-1)=1$ and $\mathrm{char}\mathbb{F}_{q^2}\neq3$.
\end{Th}

 In the above theorem, {let $\mathrm{char}\mathbb{F}_{q^2}=2$}.  Then $f(x)=x^rh(x^{q-1})$, where $r=(q+1)l+3$, $l\ge0$ is an integer
 and $h(x)=1+x+x^3$. Hence it follows from Lemma \ref{lem1} that
 the fractional polynomial
  \begin{equation}\label{eq_g1}
    g(x)=x^rh(x)^{q-1}=\frac{x^3+x^2+1}{x^3+x+1}
\end{equation}
 permutes $\mu_{q+1}$.
 Further,  we can get another class of permutation polynomial as follows.

 	As we all know, $x^2+x+1=0$ has no solution in $\mu_{q+1}$ if $k$ is even.  Then multiplying the numeractor and the denominator of $g(x)$ in
  \eqref{eq_g1} by $x^2+x+1$, we can obtain a new fractional polynomial $$\frac{x^5+x+1}{x^5+x^4+1}.$$ Thus we can get some new permutation trinomials.
  \begin{Th}
  	\label{tab1}
 	Let $q=2^k$, $h(x)=1+x^4+x^{-1}$ and $r=3+(q+1)l$. Then $f(x)=x^rh(x^{q-1})=x^{lq+l+3}+x^{(l+4)q+l-1}+x^{(l-1)q+l+4}$ is a permutation trinomial over $\mathbb{F}_{q^2}$ when $k$ is even and $\mathrm{gcd}(2l+3,q-1)=1$.
 \end{Th}

Next, the following is a special permutation trinomial in \cite{LQC} and we will generalize it.

\begin{Th}
	\cite[Theorem 4.8]{LQC}
	\label{Th2}
	Let $q=2^k$, $\mathrm{gcd}(3,k)=1$ and $f(x)=x+x^q+x^{\left.{q^2}\middle/{2}\right.-\left.{q}\middle/{2}\right.+1}$. Then $f(x)$ is a permutation trinomial over $\mathbb{F}_{q^2}$.
\end{Th}

	In Theorem \ref{Th2}, $f(x)^2=x^2+x^{2q}+x^{q^2-q+2}=x^2(1+x^{2q-2}+x^{1-q})$. Let $r=2$, $h(x)=1+x^2+x^{-1}$. Then the fractional polynomial
	$$g(x)=x^rh(x)^{q-1}=\frac{x^4+x^3+x}{x^3+x+1}$$
	 permutes $\mu_{q+1}$ when $k\not\equiv0\pmod 3$ due to  Theorem \ref{Th2} and Lemma \ref{lem1}. In return, we can get some new permutation trinomials.
	
\begin{Th}
	\label{Rem2}	 	
	 	Let $q=2^k$, where $k\not\equiv0\pmod 3$, $h(x)=1+x^2+x^{-1}$ and $r=2+(q+1)l$. Then $f(x)=x^rh(x^{q-1})=x^{lq+l+2}+x^{(l+2)q+l}+x^{(l-1)q+l+3}$ is a permutation trinomial of $\mathbb{F}_{q^2}$ where  $\mathrm{gcd}(l+1,q-1)=1$.
\end{Th}

	After multiplying the numeractor and the denominator of $g(x)$ in Theorem  \ref{Rem2} by $x^2+x+1$, we can obtain a diverse fractional polynomial  $$\frac{x^6+x^2+x}{x^5+x^4+1}$$ which permutes $\mu_{q+1}$. Let $q=2^k$, $h(x)=1+x^4+x^{-1}$ and $r=2+(q+1)l$. Then a new class of permutation trinomial is obtained by Lemma \ref{lem1}.

\begin{Th}
	\label{tab2}
		 Let $q=2^k$. Then $f(x)=x^{lq+l+2}+x^{(l+4)q+l-2}+x^{(l-1)q+l+3}$ is a permutation trinomial over $\mathbb{F}_{q^2}$ when   $k\equiv2,4\pmod 6$ and $\mathrm{gcd}(l+1,q-1)=1$.
\end{Th}

\subsection{New constructions from fractional polynomials}
\begin{Th}
	\label{Th3}
	Let $q=2^k$, where $k\not\equiv2\pmod 4$, $h(x)=1+x^3+x^{-1}$ and $r=3+(q+1)l$, where $\mathrm{gcd}(2l+3,q-1)=1$. Then $f(x)=x^rh(x^{q-1})=x^{lq+l+3}+x^{(l+3)q+l}+x^{(l-1)q+l+4}$ is a permutation trinomial over $\mathbb{F}_{q^2}$.
\end{Th}

\begin{proof}
First, we prove that $h(x)\neq0$ if $x\in\mu_{q+1}$. Suppose there exists  $a\in\mu_{q+1}$ such that $h(a)=0$. Then $a^4+a+1=0$.
Hence the order of $a$ is $15$. Combining with the assumption $a\in\mu_{q+1}$, we have  $15\mid q+1$, which contradicts with $k\not\equiv2\pmod 4$.
	
	According to Lemma \ref{lem1}, it suffcies to show that the fractional polynomial $$g(x)=x^rh(x)^{q-1}=\frac{x^5+x^4+x}{x^4+x+1}$$ permutes $\mu_{q+1}$ when $k\not\equiv2\pmod 4$.
	
	{Otherwise, } there exist $x_1, x_2\in\mu_{q+1}$ and $x_1\neq x_2$ such that $g(x_1)=g(x_2)$. Then
	\begin{equation}
	\label{3.0}
	\frac{x_1^5+x_1^4+x_1}{x_1^4+x_1+1}=\frac{x_2^5+x_2^4+x_2}{x_2^4+x_2+1}.
	\end{equation}
	Let $x_1+x_2=u$ and $x_1x_2=v$. Then $u^{1-q}=\frac{x_1+x_2}{x_1^{-1}+x_2^{-1}}=x_1x_2=v$ since $x_1,x_2\in\mu_{q+1}$. Plugging them into the above equation and simplifying, we get
	\begin{equation}
	\label{3.1}
	u^4+(v+1)u^3+vu^2+v^4+v^2+1=0.
	\end{equation}
	Let $u=y^{-1}$. Then $v=y^{q-1}$ {and} we have
	\begin{equation*}
	y^{-4}+\left(y^{q-1}+1\right)y^{-3}+y^{q-1}y^{-2}+y^{4q-4}+y^{2q-2}+1=0.
	\end{equation*}
	Let us multiply both sides of the above equation by $y^4$ and let $\alpha$ and $\beta$ denote $y+y^q$ and $y^{q+1}$ respectively. Then $\alpha,\beta\in\mathbb{F}_q$ and
	\begin{equation}
	\label{3.2}
		\beta^2+\beta+\alpha^4+\alpha+1=0.
	\end{equation}
    Eq.(\ref{3.2}) and Lemma \ref{lem6} yield $\mathrm{Tr}_q(\alpha^4+\alpha+1)=0$. However, it does not hold when $k$ is odd. Therefore, $g(x)$ permutes $\mu_{q+1}$ when $k$ is odd. In the following, we assume that $k$ is an even integer, i.e., $k\equiv0\pmod4$ since $k\not\equiv2\pmod 4$.

    Let $\omega$ be a solution of $x^3=1$ in $\mathbb{F}_q$ and $\omega\neq1$. Then

     $$ \mathrm{Tr}_q(\omega) = \underbrace{\omega+\omega^2+\cdots+\omega+\omega^2}_{k \  \text{terms}} 
     =\underbrace{1+\cdots+1}_{\frac{k}{2}} = \frac{k}{2}=0$$
      since $k\equiv0\pmod4$. Solving Eq.(\ref{3.2}), we get
    \begin{equation}
    \label{3.3}
    	\beta=\alpha^2+\alpha+\omega.
    \end{equation}

    If $y\in\mathbb{F}_q$, then $v=y^{q-1}=1$. On the other hand, $v=x_1x_2$. So $x_2=x_1^{-1}$. Plugging it into Eq.(\ref{3.0}), we obtain:
    \begin{equation*}
    \frac{x_1^5+x_1^4+x_1}{x_1^4+x_1+1}=\frac{x_1^4+x_1+1}{x_1^5+x_1^4+x_1},
    \end{equation*}
    i.e.,
    \begin{equation*}
    \frac{x_1^5+x_1^4+x_1}{x_1^4+x_1+1}=1.
    \end{equation*}
    So $x_1^5=1$. Hence $x_1^{\mathrm{gcd}(5,q+1)}=x_1=1$ since $x_1^{q+1}=1$ and $k\equiv0\pmod4$. Then $x_2=x_1^{-1}=1=x_1$, which is a contradiction.

    In the following, we assume $y\in\mathbb{F}_{q^2}\backslash \mathbb{F}_{q}$.

    We recall that $x_1$, $x_2$ are the solutions of $x^2+y^{-1}x+y^{q-1}=0$ in $\mathbb{F}_{q^2}$. Multiplying both sides of the above equation by $y^2$ and let $z=yx$, we get
    \begin{equation*}
    z^2+z+\alpha^2+\alpha+\omega=0,
    \end{equation*}
    since $y^{q+1}=\beta=\alpha^2+\alpha+\omega$.

	Let $\gamma=z+\alpha$. Then $x=\frac{\gamma+\alpha}{y}$ and
	\begin{equation}
	\label{3.4}
	\gamma^2+\gamma+\omega=0.
	\end{equation}
	Since $\mathrm{Tr}_q(\omega)=0$, Eq.(\ref{3.4}) has two solutions $\gamma_1,\gamma_2$ in $\mathbb{F}_q$. Moreover, $z_1=yx_1\in\mathbb{F}_q$, where $y=(x_1+x_2)^{-1}$. Then, $z_1^q=yx_2=yx_1$, $x_2=x_1$, which is a contradiction.
	
	This completes the proof.
\end{proof}

\begin{Th}
	\label{Th4}
	Let $q=2^k$, where $k$ is even, $h(x)=1+x^3+x^{-1}$ and $r=1+(q+1)l$, where $l\ge0$ be an integer and $\mathrm{gcd}(2l+1,q-1)=1$. Then $f(x)=x^rh(x^{q-1})=x^{lq+l+1}+x^{(l+3)q+l-2}+x^{(l-1)q+l+2}$ is a permutation trinomial over $\mathbb{F}_{q^2}$.
\end{Th}

\begin{proof}
{From the proof of Theorem \ref{Th3}, we claim that $h(x)\neq0$ when $x\in\mu_{q+1}$.  Then} we show that the fractional polynomial $$g(x)=x^rh(x)^{q-1}=\frac{x^4+x^3+1}{x^5+x^2+x}$$ permutes $\mu_{q+1}$.
	 Assume that there exist two dinstinct elements $x_1$ and $x_2$ in $\mu_{q+1}$ such that the following equation holds.
	\begin{equation}
	\frac{x_1^4+x_1^3+1}{x_1^5+x_1^2+x_1}=\frac{x_2^4+x_2^3+1}{x_2^5+x_2^2+x_2}.
	\end{equation}
	For simplicity of presentation, let $u=x_1+x_2$ and $v=x_1x_2$. Then $v=u^{1-q}$. And after a tedious computation, we yield
	\begin{equation*}
	u^4+(v+1)^3u+(v^4+v^2+1)=0.
	\end{equation*}
	Let $y=u^{-1}$. Then $u=y^{-1}$ and $v=y^{q-1}$. Plugging them into the above equation, we obtain
	\begin{equation*}
	y^{-4}+\left(y^{q-1}+1\right)^3y^{-1}+y^{4q-4}+y^{2q-2}+1=0,
	\end{equation*}
	i.e.,
	\begin{equation}
	\label{4.1}
	y^{4q}+y^4+(y^q+y)^3+y^{2q+2}+1=0,
	\end{equation}
    after multiplying both sides of the equation by $y^4$.
	
	If $y\in\mathbb{F}_q$, then $y=1$ from Eq.(\ref{4.1}). Moreover, $u=y^{-1}=1$ and $v=y^{q-1}=1$. Therefore, $x_1$ is a solution of $x^2+x+1=0$ in $\mathbb{F}_{q^2}$. It follows that $x_1^3=1$.  Hence $x_1^{\mathrm{gcd}(3,q+1)}=x_1=1$ since $x_1^{q+1}=1$ and $k$ is even. Then $x_2=x_1^{-1}=1=x_1$, which is a contradiction.
	
	If $y\in\mathbb{F}_{q^2}\backslash \mathbb{F}_{q}$, let $\alpha=y+y^q\in\mathbb{F}_q$. Then Eq.(\ref{4.1}) becomes that
	\begin{equation}
	\label{4.2}
	y^{2q+2}=\alpha^4+\alpha^3+1.
	\end{equation}
	It is clear from Eq.(\ref{4.2}) and the assumption of $\alpha$ that $y^2$ and $y^{2q}$ are the solutions of the following equation
	\begin{equation}
	\label{4.3}
	t^2+\alpha^2t+\alpha^4+\alpha^3+1=0.
	\end{equation}
	However,
	\begin{equation*}
	\mathrm{Tr}_q\left(\frac{\alpha^4+\alpha^3+1}{\alpha^4}\right)=\mathrm{Tr}_q\left(1+\frac{1}{\alpha}+\frac{1}{\alpha^4}\right)=0.
	\end{equation*}
	Therefore Eq.(\ref{4.3}) has solutions in $\mathbb{F}_q$, i.e., $y\in\mathbb{F}_q$. This is contrary to the assumption.
	
	We have thus proved the theorem by Lemma \ref{lem1}.
\end{proof}

\begin{Th}
	\label{Th5}
	Let $q=2^k$, $r=1+(q+1)l$, where $\mathrm{gcd}(2l+1,q-1)=1$, and $h(x)=1+x^4+x^{-2}$. Then $f(x)=x^rh(x^{q-1})=x^{lq+l+1}+x^{(l+4)q+l-3}+x^{(l-2)q+l+3}$ is a permutation trinomial over $\mathbb{F}_{q^2}$ if $k\not\equiv0\pmod 3$.
\end{Th}

\begin{proof}
   { It is easy to prove that $h(x)\neq0$ when $x\in\mu_{q+1}$, where $k\not\equiv0\pmod 3$. And} similarly, the key idea of the proof is to show $$g(x)=x^rh(x)^{q-1}=\frac{x^6+x^4+1}{x^7+x^3+x}$$ permutes $\mu_{q+1}$.
	We assume it does not hold. Then there exist two dinstinct elements $x_1$ and $x_2$ in $\mu_{q+1}$ such that the following equation holds,
	\begin{equation}
	\frac{x_1^6+x_1^4+1}{x_1^7+x_1^3+x_1}=\frac{x_2^6+x_2^4+1}{x_2^7+x_2^3+x_2}.
	\end{equation}
	For convenience, let $u=x_1+x_2$ and $v=x_1x_2$. Then it is easy to show $v=u^{1-q}$. And after a tedious computation, we yield
	\begin{equation}
	\label{5.0}
	u^6+(v^4+v^3+v^2+v+1)u^2+v^6+v^5+v^4+v^3+v^2+v+1=0.
	\end{equation}
	Let $y=u^{-1}$. Then $u=y^{-1}$ and $v=y^{q-1}$. Let $\alpha=y+y^q$ and $\beta=y^{q+1}$. Plugging them into the above equation, we obtain
	\begin{equation}
	\label{5.1}
	\alpha^6+(\beta+1)\alpha^4+\beta\alpha^2+\beta^3+\beta^2+1=0.
	\end{equation}
	If $y\in\mathbb{F}_q$, then $\alpha=0$ and $\beta^3+\beta^2+1=0$. Since $\beta\ne0$, let $\upsilon=\frac{1}{\beta}\ne0, 1$. Then $\upsilon\in\mathbb{F}_q$ and we have
	\begin{equation}
	\label{1}
	\upsilon^4+\upsilon^2+\upsilon=0.
	\end{equation}
	Raising the above equation to its $2$-th power, we obtain
	\begin{equation}
	\label{2}
	\upsilon^8+\upsilon^4+\upsilon^2=0.
	\end{equation}
	 Then computing Eq.(\ref{1}) $+$ Eq.(\ref{2}), we get $\upsilon^7=1$. Then $\upsilon=\upsilon^{\mathrm{gcd}(q-1,7)}=1$ since $\upsilon^{q-1}=1$ and $\mathrm{gcd}(q-1,7)=1$ when $k\not\equiv0\pmod 3$, which is a contradiction.
	
	If $y\in\mathbb{F}_{q^2}\backslash \mathbb{F}_{q}$, let $\alpha^2=\gamma+\beta+1$. Plugging it into Eq.(\ref{5.1}), we obtain
	\begin{equation}
	\label{5.2}
	\gamma^3+a\gamma+b=0,
	\end{equation}
	where $a=\beta^2+\beta+1$ and $b=\beta^3+\beta+1$.
	Since
	 \begin{eqnarray*}
	 	\frac{a^3}{b^2}&=&\frac{\left(\beta^2+\beta+1\right)^3}{\left(\beta^3+\beta+1\right)^2}\\
	 	&=&\frac{\beta^6+\beta^5+\beta^3+\beta+1}{\beta^6+\beta^2+1}\\
	 	&=& 1+\frac{\beta^5+\beta^3+\beta^2+\beta}{\beta^6+\beta^2+1}\\
	 	&=& w+w^2+\frac{\beta^2+\beta}{\beta^3+\beta+1}+\frac{\beta^4+\beta^2}{\beta^6+\beta^2+1},
	 \end{eqnarray*}
	 where $w^3=1$ and $w\neq1$. So $$\mathrm{Tr}_q\left(\frac{a^3}{b^2}+1\right)= 0. $$
    Hence, Eq.(\ref{5.2}) has three distinct solutions or no solution in $\mathbb{F}_q$ by Lemma \ref{lem5}. In the following, we will show that there exists one solution of Eq.(\ref{5.2}) that is not in $\mathbb{F}_q$. If it holds, it follows that Eq.(\ref{5.2}) has no solution in $\mathbb{F}_q$. Therefore, $g(x)$ permutes $\mu_{q+1}$.

    As we know, the quadratic derived equation of Eq.(\ref{5.2}) is
     \begin{equation}
     \label{5.21}
      t^2+bt+a^3=0.
     \end{equation}
     Let $t=bz$. Plugging it into the above equation, we get $z^2+z+\frac{a^3}{b^2}=0$.
  Due to the factorization $$\frac{a^3}{b^2}=w+w^2+\frac{\beta^2+\beta}{\beta^3+\beta+1}+\frac{\beta^4+\beta^2}{\beta^6+\beta^2+1},$$  where $w^3=1$ and $w\neq1$, therefore
   $z=w+\frac{\beta^2+\beta}{\beta^3+\beta+1}$ is a solution of $z^2+z+\frac{a^3}{b^2}=0$ . Then $t=\left(\beta^3+\beta+1\right)w+\beta^2+\beta$ is a solution of Eq.(\ref{5.21}). Moreover,
    \begin{eqnarray*}
    t&=&\left(\beta^3+\beta+1\right)w+\beta^2+\beta\\
    &=& w(\beta^3+\beta+1+w^2\beta^2+w^2\beta)\\
    &=& w(w^3\beta^3+w^2\beta^2+w\beta+1)\\
    &=& w(w\beta+1)^3,
    \end{eqnarray*}
    since $w^3=1$.

    Then $\epsilon^3=t$ has the solution $\epsilon=\sigma(w\beta+1)$, where $\sigma^3=w$, i.e., $\sigma^9=1$ and $\sigma^3\neq1$. Therefore,
    \begin{eqnarray*}
    &&\epsilon+\frac{\beta^2+\beta+1}{\epsilon}\\
    &=&\sigma(w\beta+1)+\sigma^8(w^2\beta+1)\\
    &=&(\sigma^4+\sigma^5)\beta+\sigma+\sigma^8\\
    &=& e^4\beta+e,
    \end{eqnarray*}
    where $e=\sigma+\frac{1}{\sigma}$, is a solution of Eq.(\ref{5.2}). We claim that $(e^4\beta+e)^q\ne e^4\beta+e$ when $k\not\equiv0\pmod 3$.

    Case I: $k\equiv1\pmod3$. Let $k=3l+1$. Since $\sigma^q=(\sigma^{8^l})^2=\sigma^{-2}$.  $e^q=\sigma^2+\frac{1}{\sigma^2}=e^2$. Then $(e^4\beta+e)^q=e^8\beta+e^2$. If $(e^4\beta+e)^q=e^4\beta+e$, then $e^8\beta+e^2=e^4\beta+e$. Therefore, we get
    $$\beta=\frac{e^2+e}{e^8+e^4}.$$
    Since $\beta\in\mathbb{F}_q$, we have
    $$\beta^q=\frac{e^4+e^2}{e^16+e^8}=\frac{e^2+e}{e^8+e^4},$$
    i.e., $\beta^2=\beta$, $\beta=1$. However, $(e^4\beta+e)^q=e^4\beta+e$. Then $(e^4\beta+e)^q=(e^4+e)^q=e^8+e^2=e^4+e$, i.e., $e^4+e=0$ or $1$. If $e=0$, $\sigma=1$, which is wrong. If $e=1$, $\sigma^2+\sigma+1=0$, i.e., $\sigma^3=1$. It is also wrong. Hence, in the case, $(e^4\beta+e)^q\ne e^4\beta+e$.

    Case II can be proved similarly as Case I. We omit it here.

    Hence,  Eq.(\ref{5.2}) has a solution which is not in $\mathbb{F}_q$.

	It follows that the theorem holds.
\end{proof}

\section{New permutation trinomials over $\mathbb{F}_{3^{2k}}$}
{In this section, we obtain two permutation trinomials with fixed exponents and then generalize them to two classes of permutation trinomials with one parameter in each class.}

First, {the following results about the trace  function and the norm function will  be needed later.
They can be easily verified and the proof is omitted here.}
\begin{Lemma}
	\label{lem3}
	Let $q=3^k$ and $x\in\mathbb{F}_{q^2}$. Then
	\begin{equation*}
	\mathrm{Tr}(x^2)=\mathrm{Tr}^2(x)+\mathrm{N}(x),
	\end{equation*}
	\begin{equation*}
	\mathrm{Tr}(x^5)=\mathrm{Tr}^5(x)+\mathrm{N}(x)\mathrm{Tr}^3(x)-\mathrm{N}^2(x)\mathrm{Tr}(x),
	\end{equation*}
	and
	\begin{equation*}
	\mathrm{Tr}(x^8)=\mathrm{Tr}^8(x)+\mathrm{N}(x)\mathrm{Tr}^6(x)-\mathrm{N}^2(x)\mathrm{Tr}^4(x)-\mathrm{N}^3(x)\mathrm{Tr}^2(x)-\mathrm{N}^4(x).
	\end{equation*}
\end{Lemma}

For convenience, let
$t=\mathrm{Tr}(x)$, $n=\mathrm{N}(x)$, $\alpha=\mathrm{Tr}(f(x))$ and $\beta=\mathrm{N}(f(x))$, where $x\in\mathbb{F}_{q^2}$. And suppose $s=\frac{t^2}{n}$, $\gamma=\frac{\alpha^2}{\beta}$.

\begin{Th}
	\label{3P1}
	Let $q=3^k$, where $k \not \equiv 0 \pmod 4$ and $f(x)=x-x^{2q-1}+x^{q^2-2q+2}$. Then $f(x)$ is a permutation trinomial over $\mathbb{F}_{q^2}$.
	\end{Th}
	\begin{proof}
		(i) First, we show the fact $f(\mathbb{F}_{q^2} \backslash \mathbb{F}_q) \subseteq \mathbb{F}_{q^2} \backslash \mathbb{F}_q$.
		
		If not, then there exists  $x\in \mathbb{F}_{q^2}\backslash \mathbb{F}_{q}$ such that $f(x)=f(x)^{q}$, i.e.,
		\begin{equation}
		\label{key1}
		x-x^{2q-1}+x^{3-2q}=x^q-x^{2-q}+x^{3q-2}.
		\end{equation}
		Let $y=x^{q-1}$. Then $y\neq1$ and Eq.(\ref{key1}) reduces to
		\begin{equation*}
		1-y^2+y^{-2}=y-y^{-1}+y^3.
		\end{equation*}
		i.e.,
		\begin{equation*}
		(y-1)^5=0.
		\end{equation*}
		Then $y=1$, which is a contradiction.
		
		(ii) We prove that $f(x)=c$ has at most one solution in $\mathbb{F}_{q^2}$ for any $c \in \mathbb{F}_{q^2}$.
		
		If $c \in \mathbb{F}_{q}$, we have $x \in \mathbb{F}_{q}$ by (i). Then $f(x)=x$, so we must have $x=c$.
		
		In the following, we assume $c \in \mathbb{F}_{q^2}\backslash \mathbb{F}_q$. Then it follows that
		\begin{eqnarray*}
			\mathrm{Tr}(f(x))&=&\mathrm{Tr}(x)-\mathrm{Tr}(x^{2q-1})+\mathrm{Tr}(x^{3-2q})\\
			&=&\mathrm{Tr}(x)-\frac{x^{3q}+x^3}{x^{1+q}}+\frac{x^{5q}+x^5}{x^{2+2q}}\\
			&=&\mathrm{Tr}(x)-\frac{\mathrm{Tr}(x^3)}{\mathrm{N}(x)}+\frac{\mathrm{Tr}(x^5)}{\mathrm{N}^2(x)},
			\end{eqnarray*}
			and
			\begin{eqnarray*}
				\mathrm{N}(f(x))&=&(x-x^{2q-1}+x^{3-2q})(x^q-x^{2-q}+x^{3q-2})\\
				&=&-x^{5q-3}-x^{5-3q}=-\frac{x^{8q}+x^8}{x^{3q+3}}\\
				&=&-\frac{\mathrm{Tr}(x^8)}{\mathrm{N}^3(x)}.
				\end{eqnarray*}
				Plugging Lemma \ref{lem3} into the above two equations, we have
				\begin{equation}
				\label{key3}
				\alpha=\mathrm{Tr}(f(x))=t-\frac{t^3}{n}+\frac{t^5-n^2t+nt^3}{n^2}=\frac{t^5}{n^2},
				\end{equation}
				and
				\begin{equation}
				\label{key4}
				\beta=\mathrm{N}(f(x))=-n(\frac{t^8}{n^4}+\frac{t^6}{n^3}-\frac{t^4}{n^2}-\frac{t^2}{n}-1).
				\end{equation}
				Computing $(\ref{key3})^2/(\ref{key4})$, we get
				\begin{equation}
				\label{key5}
				-\frac{s^5}{s^4+s^3-s^2-s-1}=\gamma,
				\end{equation}
				i.e.,
				\begin{equation*}
				\frac{1}{s^5}+\frac{1}{s^4}+\frac{1}{s^3}-\frac{1}{s^2}-\frac{1}{s}=\frac{1}{\gamma}.
				\end{equation*}
				Let $g(u)=u^5+u^4+u^3-u^2-u=(u-1)^5+1\in \mathbb{F}_{q}[u]$. Clearly, $g(v)=v^5$ is a permutation polynomial over $\mathbb{F}_q$ when $k\not\equiv0\pmod4$. So Eq.(\ref{key5}) has only one solution in $\mathbb{F}_q$. Plugging the solution into (\ref{key3}) and (\ref{key4}), we know that $\mathrm{Tr}(x)$ and $\mathrm{N}(x)$ are also uniquely determined by $c$. {Then the result follows from $f(\mathbb{F}_{q^2} \backslash \mathbb{F}_q) \subseteq \mathbb{F}_{q^2} \backslash \mathbb{F}_q$ and the fact that the set $\{x,x^q\}$ is uniquely determined by $\mathrm{Tr}(x)$ and $\mathrm{N}(x)$.}
				
				We complete the proof.
				\end{proof}

	In Theorem \ref{3P1}, $f(x)=x(1-x^{2(q-1)}+x^{2(1-q)})=x^rh(x^{q-1})$, where $r=1$ and $h(x)=1-x^2+x^{-2}$. According to the definition, the fractional polynomial
	\begin{equation*}
	g(x)=x^rh^{q-1}(x)=\frac{x^5+x^3-x}{-x^4+x^2+1}.
	\end{equation*}
{It follows from Lemma \ref{lem1} and Theorem \ref{3P1} that $g(x)$ permutes $\mu_{q+1}$ when $k \not \equiv 0 \pmod 4$.} And then we can obtain more permutation trinomials as follows.

\begin{Cor}
\label{Cor1}
		Let $q=3^k$, $r=1+(q+1)l$ and $h(x)=1-x^2+x^{-2}$, where $l\ge0$ and $k$ is an integer, which satisfies $k \not \equiv 0 \pmod 4$ and $\mathrm{gcd}(2l+1,q-1)=1$. Then $f(x)=x^rh(x^{q-1})=x^{lq+l+1}+x^{(l+2)q+l-1}-x^{(l-2)q+l+3}$ permutes $\mathbb{F}_{q^2}$.
\end{Cor}

\begin{Th}
	\label{3P2}
	Let $q=3^k$, where $k$ is odd and $f(x)=x+x^{3q-2}-x^{q^2-q+1}$. Then $f(x)$ is a permutation trinomial over $\mathbb{F}_{q^2}$.
\end{Th}
\begin{proof}
	The proof is similar to that in Theorem \ref{3P1}.
	
	(i) First, we show the fact $f(\mathbb{F}_{q^2} \backslash \mathbb{F}_q) \subseteq \mathbb{F}_{q^2} \backslash \mathbb{F}_q$.
	
	If not, then there exists  $x\in \mathbb{F}_{q^2}\backslash \mathbb{F}_{q}$ such that $f(x)=f(x)^{q}$. Then
	\begin{equation}
		\label{3P2.1}
		x+x^{3q-2}-x^{2-q}=x^q+x^{3-2q}-x^{2q-1}.
	\end{equation}
	Let $y=x^{q-1}$. Then Eq.(\ref{3P2.1}) becomes
	\begin{equation*}
		1+y^3-y^{-1}=y+y^{-2}-y^2.
	\end{equation*}
	i.e.,
	\begin{equation*}
		(y-1)(y^4-y^3+y^2-y+1)=0.
	\end{equation*}
	Then either $y=1$ or $y^4-y^3+y^2-y+1=0$. If $y=1$, then $x\in\mathbb{F}_q$, which is a contradiction. In the other case, let $y=z+1$. Plugging it into $y^4-y^3+y^2-y+1=0$, we get
	\begin{equation*}
		z^4+z^2-z+1=0.
	\end{equation*}
	i.e.,
		\begin{equation}
		\label{3P2.2}
		z^4=-(z+1)^2.
		\end{equation}
		Because $-1$ is a square in $\mathbb{F}_{q^2}$, there exist some $\epsilon\in\mathbb{F}_{q^2}$ such that $\epsilon^2=-1$. And $\epsilon^q=-\epsilon$ since $q=3^k$ and $k$ is odd. Then we know $z^2=\pm\epsilon(z+1)$ from Eq.(\ref{3P2.2}). Recalling $z=y-1$,  we get
		 \begin{equation}
		 \label{3P2.3}
		 y^2+(1\pm\epsilon)y+1=0.
		 \end{equation}
		 Raising Eq.(\ref{3P2.3}) to its $q$-th power, and using $y^{q+1}=1$ since $y=x^{q-1}$, we obtain
		 \begin{equation*}
		 y^{-2}+(1\mp\epsilon)y^{-1}+1=0,
		 \end{equation*}
		 i.e.,
		 \begin{equation}
		 \label{3P2.4}
		 y^2+(1\mp\epsilon)y+1=0.
		 \end{equation}
		 Computing Eq.(\ref{3P2.3})$-$Eq.(\ref{3P2.4}) leads to $y=0$. It contradicts since  $y=0$ is clearly not a solution of $y^4-y^3+y^2-y+1=0$.
    Hence, $f(\mathbb{F}_{q^2} \backslash \mathbb{F}_q) \subseteq \mathbb{F}_{q^2} \backslash \mathbb{F}_q$ holds.
	
	(ii) We prove that $f(x)=c$ has at most one solution in $\mathbb{F}_{q^2}$ for any $c \in \mathbb{F}_{q^2}$.
	
	If $c \in \mathbb{F}_{q}$, we have $x \in \mathbb{F}_{q}$ by (i). Then $f(x)=x$, so we must have $x=c$.
	
	In the following, we assume $c \in \mathbb{F}_{q^2}\backslash \mathbb{F}_q$. We have
	\begin{eqnarray*}
		\mathrm{Tr}(f(x))=\mathrm{Tr}(x)-\mathrm{Tr}(x^{2q-1})+\mathrm{Tr}(x^{3-2q}),
	\end{eqnarray*}
	which is equal to the $\mathrm{Tr}(f(x))$ in Theorem \ref{3P1},
	and
	\begin{eqnarray*}
		\mathrm{N}(f(x))&=&(x+x^{3q-2}-x^{2-q})(x^q+x^{3-2q}-x^{2q-1})\\
		&=&x^{4-2q}+x^{4q-2}-x^2-x^{2q}-x^{5q-3}-x^{5-3q}\\
		&=&-\frac{\mathrm{Tr}(x^6)}{\mathrm{N}^2(x)}-\mathrm{Tr}(x^2)-\frac{\mathrm{Tr}(x^8)}{\mathrm{N}^3(x)}.
	\end{eqnarray*}
	Plugging Lemma \ref{lem3} into the above two equations, we have
	\begin{equation}
		\label{3P2.5}
		\alpha=\mathrm{Tr}(f(x))=\frac{t^5}{n^2},
	\end{equation}
	and
	\begin{equation}
		\label{3P2.6}
		\beta=\mathrm{N}(f(x))=-n(\frac{t^8}{n^4}-\frac{t^4}{n^2}-1).
	\end{equation}
	Computing $(\ref{3P2.5})^2/(\ref{3P2.6})$, we get
	\begin{equation}
		\label{3P2.7}
		-\frac{s^5}{s^4-s^2-1}=\gamma,
	\end{equation}
	i.e.,
	\begin{equation*}
		\frac{1}{s^5}+\frac{1}{s^3}-\frac{1}{s}=\frac{1}{\gamma}.
	\end{equation*}
	Let $g(u)=u^5+u^3-u\in \mathbb{F}_{q}[u]$ {. Then according to \cite{LN}, $g$  is a normalized permutation polynomial over $\mathbb{F}_q$ if $q\equiv\pm2\pmod5$.} Recall that a  polynomial $f$ is called normalized if
it is monic, $f(0)=0$, and when the degree $n$ of $f$ is not divisible by the characteristic of $\mathbb{F}_q$, the coefficient of $x^{n-1}$ is $0$.  So Eq.(\ref{3P2.7}) has only one solution in $\mathbb{F}_q$. Plugging the solution into Eq.(\ref{3P2.5}) and Eq.(\ref{3P2.6}), we can know $\mathrm{Tr}(x)$ and $\mathrm{N}(x)$ are also uniquely determined by $c$.
	
	We complete the proof.
\end{proof}

	In Theorem \ref{3P2}, $f(x)=x(1+x^{3(q-1)}-x^{1-q})=x^rh(x^{q-1})$, where $r=1$ and $h(x)=1+x^3-x^{-1}$. Therefore, the fractional polynomial
	\begin{equation*}
	g(x)=x^rh^{q-1}(x)=\frac{-x^4+x^3+1}{x^5+x^2-x}
	\end{equation*}
	 permutes $\mu_{q+1}$ when $k$ is odd due to Lemma \ref{lem1} and Theorem \ref{3P2}. Hence, we can obtain more permutation trinomials as follows by Lemma \ref{lem1}.
\begin{Cor}
	\label{Cor2}	
	Let $q=3^k$, $r=1+(q+1)l$ and $h(x)=1+x^3-x^{-1}$, where $l\ge0$ and $k$ is an integer, which satisfies $k$ is odd and $\mathrm{gcd}(2l+1,q-1)=1$. Then $f(x)=x^rh(x^{q-1})=x^{lq+l+1}+x^{(l+3)q+l-2}-x^{(l-1)q+l+2}$ permutes $\mathbb{F}_{q^2}$.
\end{Cor}

\section{Multiplicative inequivalences of Newly constructed permutation trinomials}
 {In this section, we show that permutation trinomials in this paper are not multiplicative equivalent to the known results. First} we recall the following definition of multiplicative equivalence between permutation polynomials.
\begin{Def}
	\label{def}
	\cite{LQC}
	Two permutation polynomials $f(x)$ and $g(x)$ in $\mathbb{F}_{q}[x]$ are called multiplicative equivalent if there exists an integer $1 \le d \le q-1$ such that $\mathrm{gcd}(d,q-1)=1$ and $f(x)=g(x^d)$.
\end{Def}
As we all know, $g(x)=cf(ax+b)+d$, where $a,b,c,d\in\mathbb{F}_q$ and $a,c\ne0$ is a permutation polynomial over $\mathbb{F}_q$ if and only if  so is $f(x)$.
{Here we only consider the multiplicative equivalence since most newly constructed permutation polynomials here are with
trivial coefficients, that is, whose nonzero coefficient is only $1$.  }

\subsection{The comparison over $\mathbb{F}_{2^{2k}}$}

First, we show that our {permutation trinomials} are not multiplicative equivalent to the known results over $\mathbb{F}_{2^{2k}}$. For {convenience}, we list the known results and ours in TABLE I where  $f_i(x)$ denote the known results and $y_i(x)$ denote ours.

 \begin{table}[htp]
 	\label{table2}
 	\centering
 	\caption{Permutation trinomials of the form $x^a+x^b+x^c$ over $\mathbb{F}_{q^2}$ where $q=2^k$}
 	\begin{tabular}{c | c  c  c | c | c}
 		\hline
 		& $a$ &  $b$  & $c$ & Conditions & Resourecs\\
 		\hline
 		\multirow{5}{*}{$f_i(x)$}
 		& $1$  & $q$  & $\left.{q^2}\middle/{2}\right.-\left.{q}\middle/{2}\right.+1$ & $\mathrm{gcd}(3,k)=1$ & \cite{LQC} \\
 		& $1$ & $q$ &  $2q-1$ & {$k$ is positive} & \cite{XH3}\\
 		& $1$  & $q+2$  & $\left.{q^2}\middle/{2}\right.+\left.{q}\middle/{2}\right.+1$ & $k$ is odd & \cite{LQC} \\
 		& $1$  & $lq-(l-1)$  & $q^2-lq+l$ & {$k$ is positive}  & \cite{DQ} \\
 		& $lq+l+3$  & $(l+1)q+l+2$  & $(l+3)q+l$ & $\mathrm{gcd}(2l+3,q-1)=1$ &\cite{Zieve1} \\
 		\hline
 		\multirow{6}{*}{$y_i(x)$} 	& $lq+l+3$  & $(l+4)q+l-1$  & $(l-1)q+l+4$ & $k$ is even and $\mathrm{gcd}(2l+3,q-1)=1$  & Th \ref{tab1}\\
 			& $lq+l+2$  & $(l+2)q+l$  & $(l-1)q+l+3$ & $\mathrm{gcd}(3,k)=1$ and $\mathrm{gcd}(l+1,q-1)=1$ & Th \ref{Rem2} \\
 		& $lq+l+2$  & $(l+4)q+l-2$  & $(l-1)q+l+3$ & $k\equiv2,4\pmod 6$ and  $\mathrm{gcd}(l+1,q-1)=1$ & Th \ref{tab2}\\
 		& $lq+l+3$  & $(l+3)q+l$  & $(l-1)q+l+4$ & $k\not\equiv2\pmod 4$ and $\mathrm{gcd}(2l+3,q-1)=1$ & Th \ref{Th3}\\
 		& $lq+l+1$ &  $(l+3)q+l-2$  & $(l-1)q+l+2$ & $k$ is even and $\mathrm{gcd}(2l+1,q-1)=1$ & Th \ref{Th4}\\
 		& $lq+l+1$ & $(l+4)q+l-3$  & $(l-2)q+l+3$ & $\mathrm{gcd}(3,k)=1$ and $\mathrm{gcd}(2l+1,q-1)=1$ & Th \ref{Th5}\\
 		\hline
 	\end{tabular}
 \end{table}

Let $f_1,f_2\in\mathbb{F}_{q^2}[x]$ and $f_1(x)=x^r(1+x^{m(q-1)}+x^{n(q-1)})$. From Definition \ref{def}, if $f_1$ and $f_2$ are multiplicative equivalent, then there exists an integer $d$ such that $\mathrm{gcd}(d,q-1)=1$ and $f_2(x)=f_1(x^d)=x^{rd}(1+x^{md(q-1)}+x^{nd(q-1)})$. {Then their  corresponding fractional trinomials are as follows. }
$$g_1(x)=x^r\frac{1+x^{-m}+x^{-n}}{1+x^m+x^n},$$
$$g_2(x)=x^{rd}\frac{1+x^{-md}+x^{-nd}}{1+x^{md}+x^{nd}}.$$
Therefore, $g_2(x)=g_1(x^d)$. {Hence $f_1$ and $f_2$ are multiplicative equivalent if and only if their
corresponding fractional polynomials are equivalent. Now   let us list the newly constructed  fractional polynomials.}
$$g_1(x)=\frac{x^5+x+1}{x^5+x^4+1},\qquad g_2(x)=\frac{x^6+x^2+x}{x^5+x^4+1}.$$
$$g_3(x)=\frac{x^5+x^4+x}{x^4+x+1},\quad g_4(x)=\frac{x^4+x^3+1}{x^5+x^2+x},\quad g_5(x)=\frac{x^6+x^4+1}{x^7+x^3+x}.$$

{Then we list  the known fractional polynomials as follows.}
$$g_6(x)=\frac{x^3+x^2+1}{x^3+x+1},\qquad g_7(x)=\frac{x^4+x^3+x}{x^3+x+1}.$$
It is easy to check that these $g_i(x)$, where $i=1,2,\cdots,7$ are pairwise {multiplicative } inequivalent. Therefore, our permutation trinomials are {
inequivalent to } existing permutation trinomials.

\subsection{The comparison over $\mathbb{F}_{3^{2k}}$}
{To the best of the authors' knowledge, there is only one class of permutation trinomial in $\mathbb{F}_{3^{2k}}$, which is listed in Theorem \ref{Th}.}

\begin{Th}
	\label{Th}
	
	 \cite{XH4} Let $f=ax+bx^q+x^{2q-1}\in\mathbb{F}_{q^2}[x]$, where $q=3^k$ . Then $f$ is a permutation trinomial over $\mathbb{F}_{q^2}$ if and only if one of the following is satisfied.
		\begin{enumerate}[(i)]
			\item $\left(-a\right)^{\frac{q+1}{2}}=-1$ or $3$, $b=0$.
			\item $ab\neq0$, $a=b^{1-q}$, $1-\frac{4a}{b^2}$ is a square of $\mathbb{F}_q$.
			\item $ab(a-b^{1-q})\neq0$, $1-\frac{4a}{b^2}$ is a square of $\mathbb{F}_q$, $b^2-a^2b^{q-1}-3a=0$.
		\end{enumerate}

	\end{Th}
	{We  collected in Theorem \ref{Thh} our newly constructed classes of permutation trinomials constructed  in Section 3.}
  These functions extend the list of known such permutations.
	\begin{Th}
		\label{Thh}
	\begin{enumerate}[(1)]
		\item (Corollary \ref{Cor1}) Let $q=3^k$ and $l\ge0$, where $k \not \equiv 0 \pmod 4$ and $\mathrm{gcd}(2l+1,q-1)=1$. Then $f(x)=x^{lq+l+1}+x^{(l+2)q+l-1}-x^{(l-2)q+l+3}$ permutes $\mathbb{F}_{q^2}$.
		\item  (Corollary \ref{Cor2})  Let $q=3^k$ and $l\ge0$, where $k$ is odd and $\mathrm{gcd}(2l+1,q-1)=1$. Then $f(x)=x^{lq+l+1}+x^{(l+3)q+l-2}-x^{(l-1)q+l+2}$ permutes $\mathbb{F}_{q^2}$.
	\end{enumerate}
\end{Th}

It is easy to verify that  the corresponding fractional polynomials are inequivalent. Hence, Theorem \ref{Thh} is not multiplicative equivalent to Theorem \ref{Th}.

\subsection{Other comparison}

 Lee and Park \cite{LP} considered trinomials of the form $f=x^nh(x^{\frac{q-1}{3}})$, where $3\mid q-1$ and $h(x)=ax^2+bx+c\in\mathbb{F}_q[x]$, and they proved the following theorem

 \begin{Th}
 	\cite{LP}
 	In the above notation, $f$ is a permutation polynomial over $\mathbb{F}_q$ if and only if the following conditions are satisfied.
 	\begin{enumerate}[(i)]
 		\item  $\mathrm{gcd}\left(n,\frac{q-1}{3}\right)$=1.
 		\item  $h(\epsilon^i)\neq0$ for $0\le i<3$, where $\epsilon^3=1$.
 		\item  $\mathrm{log}_{\alpha}\frac{h(1)}{h(\epsilon)}\equiv\mathrm{log}_{\alpha}\frac{h(\epsilon)}{h(\epsilon^2)}\not\equiv n\pmod 3$, where $\alpha$ is a primitive element of $\mathbb{F}_q$ such that $\epsilon=\alpha^{\frac{q-1}{3}}$.
 	\end{enumerate}
 \end{Th}
 On one hand, if $q=2^{2k}$, $a=b=c=1$, then $f(x)=x^n(x^{\frac{2^{2k+1}-2}{3}}+x^{\frac{2^{2k}-1}{3}}+1)$. By Magma, it is easy to see that $f(x)$ does not permute $\mathbb{F}_{2^{2k}}$. On the other hand,  $\mathrm{char}\mathbb{F}_{q^2}=3$ does not satisfy the condition. Hence, our results are different to the above theorem.

 The following theorem was obtained by Zieve 	\cite{Zieve1}.

 \begin{Th}
 	\label{th}
 	\cite{Zieve1}
 	Let $r$ be a positive integer and $\beta\in\mu_{q+1}$. Let $h(x)\in\mathbb{F}_{q^2}[x]$ be a polynomial of degree $d$ such that $h(0)\neq0$ and $(x^dh(\left.{1}\middle/{x}\right.))^q=\beta h(x^q)$. Then $f(x)=x^rh(x^{q-1})$ is a permutation polynomial over $\mathbb{F}_{q^2}$ if and only if all the following hold.
 	\begin{enumerate}[(i)]
 		\item $\mathrm{gcd}(r,q-1)=1$.
 		\item $\mathrm(r-d,q+1)=1$.
 		\item $h$ has no roots in $\mu_{q+1}$.
 	\end{enumerate}
 \end{Th}
 The key point in Theorem \ref{th} is that $(x^dh(\left.{1}\middle/{x}\right.))^q=\beta h(x^q)$.
  Let $h(x)=1+x^m-x^n$, where $n>m>0$. Then $h(x^q)=1+x^{qm}-x^{qn}$ whilst $(x^dh(\left.{1}\middle/{x}\right.))^q=-1+x^{qn-qm}+x^{qn}$.
  Therefore, $h(x)$ satisfies the condition of Theorem \ref{th} if and only if $n=2m$ and $\mathrm{char}\mathbb{F}_{q^2}=2$.
   Luckly, the $h(x)$ in this paper do not belong to this situation. Hence, our results are not contained by known {t}heorems.

\section{Conclusion}

 Permutation trinomials over finite fields are both interesting and important
 in theory and in many applications. However, there does not exist an effective criteria that deals with the permutation property of a general trinomial. In this paper,
 we discover six new classes of permutation trinomials (see Table I) over $\mathbb{F}_{2^{2k}}$ {and} three new classes of permutation trinomials (see Theorem \ref{Thh}) over $\mathbb{F}_{3^{2k}}$. These functions extend the list of such permutations. Moreover, distinct from most of the known permutation trinomials which are with fixed exponents if the corresponding finite field is given, our results are  some general classes of permutation trinomials  with one parameter in the exponents. Hence they are more general and meaningful.

  With the help of a computer, some conjectures are presented as follows.
\begin{Conj}
	\label{conj1}
    \begin{enumerate}[(1)]
		 \item Let $q=3^k$, $k$ is even and $f(x)=x^{lq+l+5}+x^{(l+5)q+l}-x^{(l-1)q+l+6}$, where $\mathrm{gcd}(5+2l, q-1)=1$. Then $f(x)$ is a  permutation trinomial over $\mathbb{F}_{q^2}$.
		 \item Let $q=3^k$, $f(x)=x^{lq+l+1}-x^{(l+4)q+l-3}+x^{(l-2)q+l+3}$ and $\mathrm{gcd}(1+2l, q-1)=1$. Then $f(x)$ is a permutation trinomial over $\mathbb{F}_{q^2}$.
		 \item Let $q=3^k$, $f(x)=x^{lq+l+1}+x^{(l+2)q+l-1}-x^{(l-2)q+l+3}$ and  $\mathrm{gcd}(1+2l, q-1)=1$. Then $f(x)$ is a permutation trinomial over $\mathbb{F}_{q^2}$ if $k \not \equiv 2 \pmod 4$.
    \end{enumerate}
\end{Conj}

By Lemma \ref{lem1}, it suffices to prove the following results correspondingly.

\begin{Conj}
	\label{conj}
	\begin{enumerate}[(1)]
		\item Let $q=3^k$, $k$ is even and $g(x)=\frac{-x^7+x^3+x}{x^6+x^4-1}$. Then $g(x)$ permutes $\mu_{q+1}$.
		\item Let $q=3^k$ and $g(x)=\frac{x^6+x^4-1}{-x^7+x^3+x}$. Then $g(x)$ permutes $\mu_{q+1}$.
		\item Let $q=3^k$, $g(x)=\frac{-x^5+x^3+x}{x^4+x^2-1}$. Then $g(x)$ permutes $\mu_{q+1}$ if $k \not \equiv 2 \pmod 4$.
	\end{enumerate}
\end{Conj}
We have {verified Conjecture \ref{conj} by Magma for} $1\le k\le 6$.

{One can see from the proofs of the theorems in this paper that the key point to the proofs is to determine  the number of the solutions of some specified equations with high degree (cubic, quartic, and so on). The main difficulty to prove these conjectures lies in dealing with these high degree equations.
}

\end{document}